\documentclass[11pt]{article}
\usepackage[T1]{fontenc}
\usepackage[utf8]{inputenc}
\usepackage{lmodern}
\usepackage{amsmath,amssymb,amsthm,mathtools}
\usepackage{graphicx,booktabs,microtype,float}
\usepackage[margin=1in]{geometry}
\usepackage{cite}
\usepackage[hidelinks]{hyperref}
\hypersetup{pdftitle={Extrinsic MANOVA for Oriented Projective Shapes},pdfauthor={Musab Alamoudi, Robert L. Paige, Victor Patrangenaru}}
\newcommand{\R}{\mathbb R}
\newcommand{\Ss}{\mathbb S}
\newcommand{\OPS}{\mathrm{OPS}}
\newcommand{\VW}{\mathrm{VW}}
\newcommand{\norm}[1]{\lVert #1\rVert}
\newcommand{\cov}{\operatorname{Cov}}
\newcommand{\diag}{\operatorname{diag}}
\newcommand{\col}{\operatorname{col}}
\newcommand{\tr}{\operatorname{tr}}
\newcommand{\sgn}{\operatorname{sign}}
\newcommand{\Imq}{\operatorname{Im}}
\newtheorem{theorem}{Theorem}[section]
\newtheorem{proposition}[theorem]{Proposition}
\newtheorem{corollary}[theorem]{Corollary}
\theoremstyle{remark}
\newtheorem{remark}[theorem]{Remark}
\numberwithin{equation}{section}
\title{Extrinsic MANOVA for Oriented Projective Shapes:\\ Covariance Structure, Nonparametric Bootstrap Inference, and Facial Data}
\author{Musab Alamoudi\thanks{Department of Statistics, Florida State University, Tallahassee, FL, USA.}
\and Robert L.\ Paige\thanks{Department of Mathematics and Statistics, Missouri University of Science and Technology, Rolla, MO, USA.}
\and Victor Patrangenaru\thanks{Department of Statistics, Florida State University, Tallahassee, FL, USA.}}
\date{October 2026}
\begin{document}
\maketitle
\begin{abstract}
For configurations of landmarks reconstructed from camera images, classical projective shape analysis identifies antipodal spherical representatives of projective coordinates, while oriented projective shape (OPS) retains their distinction. In 3D, a fixed oriented-projective coordinate representation of a OPS is point on product of 3D spheres. We develop nonparametric multisample inference for equality of population extrinsic means on this product. Quaternion multiplication helps create a global orthonormal tangent frame field, in which we obtain explicit extrinsic covariance blocks that preserve dependence between registered landmarks. A covariance-weighted MANOVA statistic has a chi-square null limit under nonfocality of distributions and positive-definite extrinsic covariance assumptions. An inverse-free discrepancy has a weighted chi-square limit when covariance matrices may be singular, and centered, groupwise nonparametric bootstrap procedures are justified asymptotically. We apply a fixed-observed-covariance nonparametric bootstrap version of the weighted statistic to 31 seven-landmark facial reconstructions from five subjects. Separate tests at the right inner brow and chin reject equality of the five OPS mean directions after Bonferroni adjustment. The analysis concerns a common mixed-sign frame domain and marginal landmarks; it does not establish pairwise subject differences or a joint six-dimensional facial test.
\end{abstract}
\section{Introduction}
Digital-camera reconstructions represent projective shapes of labeled landmark configurations, but projective registration and statistical comparison require a careful distinction between axes and directions. Classical projective shape (PS) identifies nonzero vectors that differ by arbitrary nonzero scale \cite{GoMa,MaPa:05}; oriented projective shape (OPS) identifies only projective shapes of landmark configuations seen in front of the camera. Thus, a pair of opposite unit vectors represents one PS axis but two OPS directions. The geometric formulation and examples of OPS were developed in \cite{ChZhHePa}; asymptotic inference for extrinsic means on manifolds is treated in \cite{BhPa:2005,PatElling}. Extrinsic total-variance methodology for oriented projective shapes was developed for planar images in \cite{AlPaPa} and for three-dimensional structure-from-motion reconstructions in \cite{AlPaPaSfM}; the present article addresses the comparison of extrinsic means across groups.

The objective of this article is to test equality of population \emph{extrinsic means of oriented projective shapes} across independent groups. The natural representation for three-dimensional configurations with $k$ landmarks and a specified signed registration is $(\Ss^3)^{k-5}$. That representation has two features useful for multisample inference: its inclusion in an ambient Euclidean space makes mean directions explicit, and its unit-quaternion structure provides a smooth global tangent basis. Dependence between different registered landmarks is retained in the covariance construction, rather than removed by assuming a block-diagonal covariance.

Projective mean-shape inference from digital-camera observations provides related background \cite{PaQiBu:2014}. We derive two distinct multisample procedures. First, a group-covariance-weighted quadratic form fits an unknown common tangent coordinate, yielding a chi-square limit when the required full covariance matrices are positive definite. Second, an inverse-free quadratic discrepancy can be used when sample covariance matrices cannot be inverted; its limiting distribution is a weighted sum of independent chi-square variables, not the first test's chi-square law. We also establish nonparametric bootstrap approximations and describe the fixed-observed-covariance implementation used in the facial application.

The application concerns five groups of independent facial reconstructions described in \cite[Section~12.5]{PaOs}. It tests the equality of means for the OPS coordinates of two landmark with respect to the oriented projective frame made by the other five labeled landmarks. We compare the resulting inference  with a parallel PS analysis based on the Veronese--Whitney (VW) embedding for the same configuration, and reports quaternion changes between estimated group means and the overall mean.

\section{Signed Projective Coordinates and Quaternion Geometry}
\subsection{A Common Coordinate Domain}
For a configuration of $k\ge6$ labeled points in $\R^3$, write $x_1,\ldots,x_k\in\R^4\setminus\{0\}$ for fixed positively scaled homogeneous lifts. Relabel the ordered five frame landmarks as $1,\ldots,5$ within this construction. Let
\begin{equation}\label{eq:frame}
 U=(x_1\ x_2\ x_3\ x_4),\qquad \boldsymbol{\alpha}=U^{-1}x_5,\qquad C=U\diag(\boldsymbol{\alpha}),\qquad
 z_s=\frac{C^{-1}x_s}{\norm{C^{-1}x_s}},\quad s=6,\ldots,k.
\end{equation}
We assume that the basis matrix is nonsingular and that every component $\alpha_j$ of $\boldsymbol{\alpha}$ is nonzero. Specify a domain by fixing the determinant sign $\delta=\sgn\det U$ and the coefficient-sign vector $\sigma=(\sgn\alpha_1,\ldots,\sgn\alpha_4)$. The frame coefficient vector $\boldsymbol{\alpha}$ is unrelated to the scalar significance level $\alpha$ used below. The positive-coefficient oriented frame in \cite[Proposition~3.3]{ChZhHePa} requires positive basis orientation and positive coefficients. Here we use a separately justified fixed-sign construction: mixed coefficient signs are permitted, but the prescribed signature must be the same throughout the statistical population under consideration. We do not assert that a mixed-sign frame satisfies the positive-coefficient definition of the cited paper.

The transformations acting on oriented rays form the group
\[
 \mathcal G_+=GL^+(4,\R)/\{cI_4:c>0\}
\]
of positive-determinant matrices modulo positive scalars. The quotient is by positive scalars only: a matrix and its negative are identified in the usual projective group, yet they act differently on oriented rays, and in the even dimension $4$ the determinant sign does not separate them.

\begin{proposition}[Fixed-sign spherical representation]\label{prop:fixed}
On a prescribed nondegenerate fixed-sign domain $(\delta,\sigma)$, the tuple $z=(z_6,\ldots,z_k)\in(\Ss^3)^{k-5}$ of \eqref{eq:frame} is invariant under the action of $\mathcal G_+$ and positive rescaling of every lift. Subject to no additional restrictions on the remaining landmarks, it determines the oriented configuration modulo these transformations, and gives a smooth identification of the fixed-sign quotient with $(\Ss^3)^{k-5}$.
\end{proposition}
\begin{proof}
For $x_j'=c_jAx_j$ with $c_j>0$ and $\det A>0$, direct calculation gives $\alpha_j'=c_5\alpha_j/c_j$ and $C'=c_5 AC$. Consequently $(C')^{-1}x_s'=(c_s/c_5)C^{-1}x_s$, proving invariance of each signed unit coordinate and of the signature. Conversely, if two configurations have the same signature and all their signed coordinates agree, $H=C'C^{-1}$ has positive determinant because $\sgn\det C=\delta\prod_{j=1}^4\sigma_j$ is fixed. Moreover, $Hx_j=(\alpha_j'/\alpha_j)x_j'$ for $j\le4$, $Hx_5=x_5'$, and $Hx_s$ is a positive multiple of $x_s'$ for $s\ge6$. Thus the configurations lie in one $\mathcal G_+$-orbit. A smooth representative for any $(z_6,\ldots,z_k)$ is obtained with $U_0=\diag(\delta,1,1,1)$, $\boldsymbol{\alpha}_0=\sigma$, $x_5=U_0\sigma$, and $x_s=U_0\diag(\sigma)z_s$. Inversion and normalization are smooth on the nondegenerate domain.
\end{proof}

Put $q=k-5$ and $\mathcal M=(\Ss^3)^q$, whose factor $s=1,\ldots,q$ is the coordinate $z_{s+5}$, and let $j:\mathcal M\hookrightarrow\R^{4q}$ be inclusion. The projective frame in \eqref{eq:frame} is different from the tangent basis used below. An unspecified change of projective frame need not preserve extrinsic distances, even though it leaves the underlying projective orbit unchanged.

\subsection{Unit Quaternions and a Tangent Frame}
Represent a quaternion by a scalar-last vector $h=(\xi,t)\in\R^3\times\R$, with product $(\xi,t)\odot(\eta,u)=(t\eta+u\xi+\xi\times\eta,\ tu-\xi^\top\eta)$, conjugate $\bar h=(-\xi,t)$, and imaginary part $\Imq(h)=\xi$. Unit quaternions form $\Ss^3$, with inverse $h^{-1}=\bar h$, and right multiplication by a unit quaternion is an orthogonal linear transformation of $\R^4$. For $v=(x,y,z,t)^\top\in\Ss^3$, the images of the three imaginary basis vectors under right multiplication by $v$ give
\begin{equation}\label{eq:B}
 B(v)=\begin{pmatrix}
 t&z&-y\\ -z&t&x\\ y&-x&t\\ -x&-y&-z
 \end{pmatrix}\in\R^{4\times3},\qquad
 B(v)^\top v=0,\quad B(v)^\top B(v)=I_3,\quad B(v)B(v)^\top=I_4-vv^\top.
\end{equation}
Since right multiplication preserves inner products, the columns of $B(v)$ are an orthonormal basis of $T_v\Ss^3$; their dependence on $v$ is smooth. For $v=(v^1,\ldots,v^q)\in\mathcal M$, use the block diagonal $4q\times3q$ matrix
\[
 B(v)=\diag\bigl(B(v^1),\ldots,B(v^q)\bigr),
 \qquad \operatorname{tan}_v(h)=B(v)^\top h.
\]
The tangent coordinates admit the alternative quaternion form $B(v)^\top h=\Imq(h\odot v^{-1})$ in each factor.

\section{Extrinsic Means and Covariance Matrices}
For groups $a=1,\ldots,g$, let $Y_{a,1},\ldots,Y_{a,n_a}$ be i.i.d. random elements of $\mathcal M$ with group distribution $Q_a$; assume independent groups and $n_a/n\to\lambda_a\in(0,1)$, where $n=\sum_a n_a$ and $\sum_a\lambda_a=1$. Dependence among the $q$ coordinates within an observation is allowed. Let $Y_{a,i}^s\in\Ss^3$ denote the $s$th factor. Define the population ambient means, mean resultant lengths, and extrinsic means by
\begin{equation}\label{eq:popmeans}
 \mu_a^s=E(Y_{a,1}^s),\quad \rho_a^s=\norm{\mu_a^s}>0,\quad
 \mu_{a,E}^s=\mu_a^s/\rho_a^s,
 \qquad \mu_{a,E}=(\mu_{a,E}^1,\ldots,\mu_{a,E}^q).
\end{equation}
For the sample use $\bar Y_{a,n_a}^s=n_a^{-1}\sum_iY_{a,i}^s$, $R_{a,n_a}^s=\norm{\bar Y_{a,n_a}^s}$, and $\bar Y_{a,E,n_a}^s=\bar Y_{a,n_a}^s/R_{a,n_a}^s$ when the observed resultant is nonzero. Positivity of every $\rho_a^s$ is the nonfocality condition for extrinsic means on $\Ss^3$.

The differential of the normalization $P_j(x)=x/\norm{x}$ at $\mu_a^s$ is $(I_4-\mu_{a,E}^s(\mu_{a,E}^s)^\top)/\rho_a^s$. For cross-covariance blocks $\Sigma_a^{st}=\cov(Y_{a,1}^s,Y_{a,1}^t)\in\R^{4\times4}$, the extrinsic tangent covariance has the explicit blocks
\begin{equation}\label{eq:covblock}
 \Sigma_{a,E}^{st}=
 \frac{B(\mu_{a,E}^s)^\top\Sigma_a^{st}B(\mu_{a,E}^t)}{\rho_a^s\rho_a^t},
 \quad 1\le s,t\le q,
 \qquad \Sigma_{a,E}=[\Sigma_{a,E}^{st}]_{s,t=1}^q\in\R^{3q\times3q}.
\end{equation}
The off-diagonal blocks retain within-observation dependence. They are not between-group covariances.

Write $\Sigma_a=[\Sigma_a^{st}]_{s,t=1}^q$ for the full ambient covariance and set $A_a=\diag(A_a^1,\ldots,A_a^q)$, where $A_a^s=(I_4-\mu_{a,E}^s(\mu_{a,E}^s)^\top)/\rho_a^s$. Then $\Sigma_{a,\mu}=A_a\Sigma_aA_a^\top$ is the limiting covariance of $\sqrt{n_a}(\bar Y_{a,E,n_a}-\mu_{a,E})$, and $G_a(v)=B(v)^\top\Sigma_{a,\mu}B(v)$ is its expression in the tangent basis at $v$; in particular, $\Sigma_{a,E}=G_a(\mu_{a,E})$. For computation, set $S_{a,n_a}=n_a^{-1}\sum_i(Y_{a,i}-\bar Y_{a,n_a})(Y_{a,i}-\bar Y_{a,n_a})^\top$ and let $\widehat A_{a,n_a}$ be the block-diagonal differential evaluated at the group ambient sample mean, with blocks $(I_4-\bar Y_{a,E,n_a}^s(\bar Y_{a,E,n_a}^s)^\top)/R_{a,n_a}^s$. At a common tangent reference $v\in\mathcal M$, define
\begin{equation}\label{eq:G}
 G_{a,n_a}(v)=B(v)^\top\widehat A_{a,n_a}S_{a,n_a}
 \widehat A_{a,n_a}^\top B(v)\in\R^{3q\times3q}.
\end{equation}
The differential is estimated at the \emph{group} ambient mean, while the tangent basis is evaluated at the chosen \emph{common reference}; the two evaluation points play different roles.

\begin{proposition}[Extrinsic central limit theorem]\label{prop:clt}
Under the preceding sampling conditions and $\rho_a^s>0$ for every group and coordinate, the extrinsic sample means are strongly consistent, the scaled errors $\sqrt{n_a}(\bar Y_{a,E,n_a}-\mu_{a,E})$ have the normal limit $\mathcal N_{4q}(0,\Sigma_{a,\mu})$, and in the quaternion basis
\begin{equation}\label{eq:clt}
 \sqrt{n_a}\,B(\mu_{a,E})^\top(\bar Y_{a,E,n_a}-\mu_{a,E})
 \ \xrightarrow{d}\ \mathcal N_{3q}(0,\Sigma_{a,E}).
\end{equation}
The covariance estimates $G_{a,n_a}(v_n)$ converge in probability to $G_a(v)$ whenever $v_n\to v$ in probability.
\end{proposition}
\begin{proof}
Every $Y_{a,i}$ has fixed squared norm $q$, so all required moments exist. Apply the strong law and multivariate central limit theorem to the ambient sample means, followed by the continuous mapping theorem and multivariate delta method for blockwise normalization. The derivative is block diagonal, giving \eqref{eq:covblock}. Empirical covariance consistency and continuity of the derivative and $B$ give the assertion for \eqref{eq:G}; no invertibility is needed for this proposition.
\end{proof}

\section{A Covariance-Weighted OPS MANOVA Statistic}
The hypothesis of interest is
\begin{equation}\label{eq:null}
 H_0:\mu_{1,E}=\cdots=\mu_{g,E}=\mu_E,\qquad
 H_1:\text{at least two population extrinsic means differ}.
\end{equation}
The common value $\mu_E$ is unknown. Neither the population covariance matrices nor the entire group distributions are assumed equal. Normalize the weighted average of the group extrinsic sample means componentwise to obtain the pooled reference
\begin{equation}\label{eq:pool}
 \widehat\theta_n=\sum_{a=1}^g\frac{n_a}{n}\bar Y_{a,E,n_a},\qquad
 v_n=P_j(\widehat\theta_n),\qquad D_n=B(v_n),
\end{equation}
where $P_j$ normalizes each nonzero $\R^4$ block. The pooled reference is formed from the group \emph{extrinsic sample means}, rather than by pooling the original ambient observations before normalization. Under $H_0$, $v_n\to\mu_E$ in probability. Define the $3q$-dimensional discrepancies and covariance estimates
\begin{equation}\label{eq:disc}
 u_{a,n}=D_n^\top(\bar Y_{a,E,n_a}-v_n),\qquad
 G_{a,n}=G_{a,n_a}(v_n),\qquad W_{a,n}=n_aG_{a,n}^{-1}.
\end{equation}
The exact identity $\sum_an_au_{a,n}=0$ follows by taking the tangent projection of the pooled numerator. It does \emph{not} imply $\sum_aW_{a,n}u_{a,n}=0$ when covariances differ. Consequently, fitting the common tangent coordinate is essential:
\begin{equation}\label{eq:weighted}
 \widehat c_n=\left(\sum_aW_{a,n}\right)^{-1}\sum_aW_{a,n}u_{a,n},\qquad
 \boxed{T_n^{\OPS}=\sum_{a=1}^g(u_{a,n}-\widehat c_n)^\top W_{a,n}(u_{a,n}-\widehat c_n)}.
\end{equation}
Equivalently $T_n^{\OPS}=\min_{c\in\R^{3q}}\sum_a(u_{a,n}-c)^\top W_{a,n}(u_{a,n}-c)$. The test uses full $3q\times3q$ covariance matrices, including all off-diagonal blocks in \eqref{eq:covblock}.

\begin{theorem}[OPS MANOVA null distribution]\label{thm:wald}
Assume $H_0$ in \eqref{eq:null}, independent groups with limiting proportions $\lambda_a\in(0,1)$, positive population resultants, and positive definiteness of every full $\Sigma_{a,E}$. Then
\[
 T_n^{\OPS}\xrightarrow{d}\chi^2_{3q(g-1)}.
\]
Thus comparison with the upper chi-square $(1-\alpha)$ quantile gives a test of asymptotic level $\alpha$.
\end{theorem}
\begin{proof}
Write $p=3q$ and use the population common frame $B_0=B(\mu_E)$. Translation of all discrepancy vectors by the same reference-error vector leaves their minimized quadratic form unchanged. Since $D_n-B_0=o_p(1)$, replacing the reference-based group discrepancies by $x_{a,n}=B_0^\top(\bar Y_{a,E,n_a}-\mu_E)$ changes the statistic by $o_p(1)$. The independent scaled group errors $\sqrt n\,x_{a,n}$ jointly converge to a normal vector with block covariance $V=\diag(\Sigma_{1,E}/\lambda_1,\ldots,\Sigma_{g,E}/\lambda_g)$. The block-diagonal estimator $\widehat V_n=\diag\bigl((n/n_a)G_{a,n}\bigr)_{a=1}^{g}$ satisfies $\widehat V_n\xrightarrow{p}V$, so the empirical precision matrix $\widehat V_n^{-1}$ converges to $V^{-1}$. Put $J=\col(I_p,\ldots,I_p)\in\R^{gp\times p}$. The limiting quadratic form is
\[
 \min_{c\in\R^p}(Z-Jc)^\top V^{-1}(Z-Jc),\qquad Z\sim\mathcal N_{gp}(0,V).
\]
After whitening, this is the squared norm of a standard normal vector projected orthogonally onto the complement of the $p$-dimensional column space of $V^{-1/2}J$. That projection has rank $gp-p=3q(g-1)$, proving the stated chi-square law. Consistency of the sample covariance inverses follows from the positive definite limits.
\end{proof}

\begin{corollary}[Two populations]\label{cor:two}
For $g=2$, whenever the covariance estimates are positive definite,
\[
 T_n^{\OPS}=(u_{1,n}-u_{2,n})^\top
 \left(\frac{G_{1,n}}{n_1}+\frac{G_{2,n}}{n_2}\right)^{-1}
 (u_{1,n}-u_{2,n}),
\]
and its limiting null law is $\chi^2_{3q}$ under Theorem~\ref{thm:wald}.
\end{corollary}

\begin{remark}[Scope under fixed alternatives]\label{rem:antipodal}
Theorem~\ref{thm:wald} is a null-distribution statement; it does not make $T_n^{\OPS}$ consistent against every fixed alternative. Since $B(v)^\top v=0$, a component of a mean difference along the reference direction is invisible in tangent coordinates: two populations with antipodal extrinsic means, $\mu_{1,E}^s=-\mu_{2,E}^s$, have the same tangent coordinate at a reference along that axis, and with equal limiting proportions the pooled normalization itself fails on that block. The inverse-free statistic of Section~\ref{sec:direct} separates every fixed alternative without assuming that the means lie in a common hemisphere (Theorem~\ref{thm:direct}). In the facial application all observed quaternion changes between group means have positive real part (Section~\ref{sec:changes}), so the estimated mean directions are far from this configuration.
\end{remark}

\section{Singular Covariance Matrices and Nonparametric Bootstrap Inference}
\subsection{A Direct Inverse-Free Comparison}\label{sec:direct}
Each estimate $G_{a,n_a}(v)$ is the average of $n_a$ outer products of the vectors $B(v)^\top\widehat A_{a,n_a}(Y_{a,i}-\bar Y_{a,n_a})$, which sum to zero; hence
\begin{equation}\label{eq:rank}
 \operatorname{rank}G_{a,n_a}(v)\le\min(3q,n_a-1).
\end{equation}
If $n_a-1<3q$, the group covariance is structurally singular, even if the population covariance is positive definite. A separate statistic avoids covariance inversion. Define
\begin{equation}\label{eq:direct}
 Q_n^{\OPS}=\sum_an_a\norm{\bar Y_{a,E,n_a}-\widehat\theta_n}^2
 =\frac1n\sum_{a<b}n_an_b\norm{\bar Y_{a,E,n_a}-\bar Y_{b,E,n_b}}^2.
\end{equation}
It is a discrepancy among \emph{group extrinsic means}, not an ordinary within-group total-variance statistic.

\begin{theorem}[Inverse-free null limit]\label{thm:direct}
Assume positive population resultants and the independent-group sampling conditions. Under $H_0$, set $N=4q$, $\Sigma_{a,\mu}=A_a\Sigma_aA_a^\top$, $\Omega=\diag(\Sigma_{1,\mu},\ldots,\Sigma_{g,\mu})$, $w=(\sqrt{\lambda_1},\ldots,\sqrt{\lambda_g})^\top$, and $\mathcal L=(I_g-ww^\top)\otimes I_N$. If $\tr(\mathcal L\Omega)>0$, then
\begin{equation}\label{eq:directlaw}
 Q_n^{\OPS}\xrightarrow{d}\sum_{h=1}^{gN}\omega_hZ_h^2,
\end{equation}
where $Z_h$ are independent standard normal variables and $\omega_h\ge0$ are the eigenvalues of $\Omega^{1/2}\mathcal L\Omega^{1/2}$. Under a fixed alternative of unequal population extrinsic means, $Q_n^{\OPS}/n$ converges almost surely to a positive constant.
\end{theorem}
\begin{proof}
Under $H_0$, $\epsilon_{a,n}=\sqrt{n_a}(\bar Y_{a,E,n_a}-\mu_E)$ converge jointly to centered independent normal columns with covariances $\Sigma_{a,\mu}$. Weighted centering shows that $Q_n^{\OPS}=\epsilon_n^\top\mathcal L_n\epsilon_n$ exactly, where $\mathcal L_n=(I_g-w_nw_n^\top)\otimes I_N$ and $w_n=(\sqrt{n_1/n},\ldots,\sqrt{n_g/n})^\top$. Since $\mathcal L_n\to\mathcal L$, the continuous mapping theorem and diagonalization yield \eqref{eq:directlaw}. Under an alternative, consistency of the extrinsic group means makes $Q_n^{\OPS}/n$ converge to $\sum_a\lambda_a\norm{\mu_{a,E}-\sum_b\lambda_b\mu_{b,E}}^2$, positive unless all means coincide.
\end{proof}

The null law in \eqref{eq:directlaw} is not generally $\chi^2_{3q(g-1)}$; its weights $\omega_h$ depend on the unknown covariances, so critical values are obtained from the bootstrap approximation below. The facial results do \emph{not} use $Q_n^{\OPS}$.

\subsection{Nonparametric Bootstrap for Independent Groups}
Following nonparametric bootstrap principles \cite{Efron:79,KoBaHaPa:2015}, resample \emph{complete} configurations with replacement independently within each group, retaining any empirical dependence among their coordinate landmarks. Let $\bar Y_{a,E,n_a}^*$ be the bootstrap group extrinsic mean and $e_a^*=\bar Y_{a,E,n_a}^*-\bar Y_{a,E,n_a}$. Centering the resampled group means at their own observed means removes the original between-group differences under the bootstrap null.

For the direct statistic, define $\overline e_n^*=\sum_a(n_a/n)e_a^*$ and
\[
 Q_n^{\OPS,*}=\sum_an_a\norm{e_a^*-\overline e_n^*}^2.
\]
For the covariance-weighted statistic, put $e_{a,T}^*=D_n^\top e_a^*$. A \emph{studentized bootstrap} recomputes the group covariance estimates in each resample and fits the common coordinate again. In the \emph{fixed-observed-covariance bootstrap} used in the facial application, keep $D_n$, $G_{a,n}$, and $W_{a,n}$ fixed and use
\begin{equation}\label{eq:fixedboot}
 \widehat c_{F,n}^*=\left(\sum_aW_{a,n}\right)^{-1}\sum_aW_{a,n}e_{a,T}^*,\qquad
 T_{F,n}^{\OPS,*}=\sum_a(e_{a,T}^*-\widehat c_{F,n}^*)^\top W_{a,n}(e_{a,T}^*-\widehat c_{F,n}^*).
\end{equation}
The fitted common coordinate is required in both the observed and resampled statistics.

\begin{proposition}[First-order bootstrap validity]\label{prop:boot}
Under $H_0$ and the nondegeneracy condition of Theorem~\ref{thm:direct}, the conditional distribution of $Q_n^{\OPS,*}$ converges in probability to the law in \eqref{eq:directlaw}. Under the stronger conditions of Theorem~\ref{thm:wald}, the conditional distribution of the centered studentized or fixed-observed-covariance statistic converges in probability to $\chi^2_{3q(g-1)}$. Thus their corresponding conditional upper quantiles give tests of asymptotic level $\alpha$.
\end{proposition}
\begin{proof}
The centered groupwise empirical bootstrap has the conditional multivariate CLT for ambient sample means. Positive population resultants ensure that the componentwise normalization is smooth in neighborhoods containing the observed and resampled means with probability tending to one. The conditional delta method gives $\sqrt{n_a}\,e_a^*=\widehat A_{a,n_a}\sqrt{n_a}(\bar Y_{a,n_a}^*-\bar Y_{a,n_a})+o_{p^*}(1)$, with covariance limits $\Sigma_{a,\mu}$. Conditional independence across groups gives the same joint Gaussian limit as under $H_0$. Weighted centering and the continuous mapping theorem yield the direct limit. Under the stronger positive-definiteness assumptions, consistent group covariance estimates and the common tangent frame yield the same whitened fitted quadratic form as in Theorem~\ref{thm:wald}; replacing resampled covariance estimates by consistent fixed observed estimates does not alter this first-order limit. Continuity of the nondegenerate limiting distribution at its upper quantile completes the argument.
\end{proof}

\begin{remark}[Finite-sample scope]
These bootstrap statements are first-order asymptotic. In a small group, resampling may produce singular empirical covariance matrices: in the facial data of Section~\ref{sec:faces}, recomputing the $3\times3$ covariance estimates within resamples gives a singular matrix in some group for $14{,}329$ of $20{,}000$ resamples, exactly those in which a group draws at most three distinct reconstructions. Keeping the original positive-definite covariance estimates fixed avoids new inversions, but does not create exact finite-sample type-I-error control. An ordinary label-permutation test would require an exchangeability hypothesis stronger than equality of population extrinsic means \cite{ChRo:2013}. No covariance regularization or pseudoinverse is used here.
\end{remark}

\section{Five-Group Facial Application}\label{sec:faces}
\subsection{Data, Registration, and Test Design}
We use the data described in \cite[Section~12.5]{PaOs}: $31$ separate reconstructions of seven 3D facial landmarks from five subjects. The ordered subject identifiers are M006, M007, M008, M004, and M005, with sizes $(6,6,7,6,6)$ (Table~\ref{tab:subjects}). One complete reconstruction is the sampling unit. The statistical model assumes independent groups and independent, identically distributed reconstructions within each group.

\begin{figure}[H]
\centering
\includegraphics[width=\textwidth]{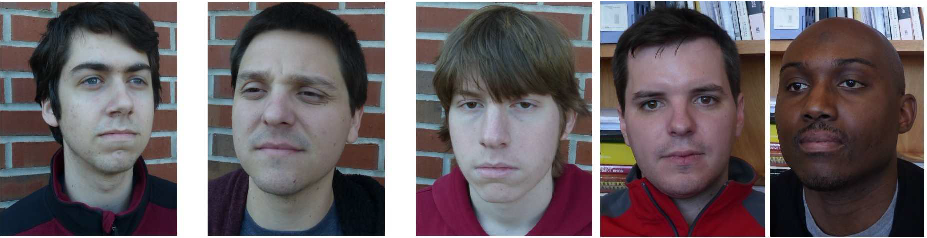}\\[5pt]
\begin{tabular}{@{}ccccc@{}}
\includegraphics[height=1.23in,keepaspectratio]{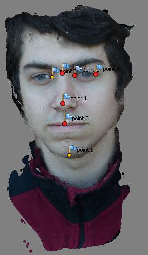} &
\includegraphics[height=1.23in,keepaspectratio]{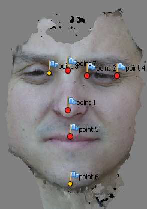} &
\includegraphics[height=1.23in,keepaspectratio]{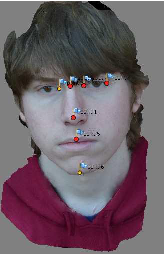} &
\includegraphics[height=1.23in,keepaspectratio]{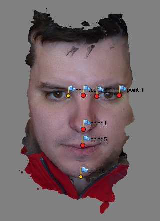} &
\includegraphics[height=1.23in,keepaspectratio]{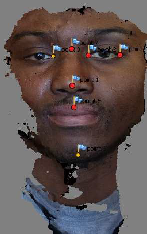}
\end{tabular}
\caption{Photographs (top) and seven-landmark facial reconstructions (bottom) of Subjects 1--5, in the order used for the analysis. The displayed reconstructions illustrate landmark placement; the calculations use the complete recorded configurations. Images from Patrangenaru and Osborne \cite[Section~12.5, Figures~12.1--12.2]{PaOs}.}
\label{fig:faces}
\end{figure}

The registered projective frame has basis landmarks $1,2,4,7$ and normalizing landmark $5$ (lip). The two retained coordinates are original landmark $L=3$ (right inner brow) and $L=6$ (chin). For Cartesian landmark positions $p_L\in\R^3$, put $x_L=(p_L^\top,1)^\top$, and in each reconstruction define
\begin{equation}\label{eq:face}
 U=(x_1\ x_2\ x_4\ x_7),\quad \boldsymbol{\alpha}=U^{-1}x_5,\quad C=U\diag(\boldsymbol{\alpha}),\quad
 Y^{(L)}=\frac{C^{-1}x_L}{\norm{C^{-1}x_L}},\quad L\in\{3,6\}.
\end{equation}
All 31 observed frames have $\det U>0$ and the constant \emph{mixed} coefficient signature $(+,+,-,-)$. Among the $105$ ways to choose four basis landmarks and a unit landmark from the seven recorded landmarks, none satisfies the all-positive coefficient condition in all $31$ reconstructions; the largest coverage is $28$. Thus no frame formed from the recorded landmarks meets this condition for the full sample. The application therefore uses Proposition~\ref{prop:fixed}, which covers all $31$ reconstructions under the observed signature, rather than the all-positive frame condition from \cite{ChZhHePa}. No observation is reflected to align it with a group mean. The validity of this signed-coordinate interpretation is conditional on the fixed lift and frame-domain model.

The aim is to test equality of the five population OPS mean directions \emph{at each landmark separately}: for $L\in\{3,6\}$, $H_{0,L}:\mu_{1,E}^{(L)}=\cdots=\mu_{5,E}^{(L)}$. Here $q=1$, the single factor is $Y^{(L)}$, the tangent dimension is three, and five group-specific $3\times3$ covariance estimates are required. The joint two-landmark model has $q=2$ and tangent dimension six. By \eqref{eq:rank}, each of the four groups with six reconstructions has joint sample covariance rank at most five; hence the inverse-based six-dimensional joint weighted test is unavailable without further methodology. We therefore use two \emph{marginal} tests and control their familywise error by Bonferroni with per-landmark level $0.025$.

\begin{table}[t]
\centering
\caption{Facial reconstruction groups and their correspondence to original data identifiers. Each observation comprises seven labeled 3D landmarks.}
\label{tab:subjects}
\begin{tabular}{lccccc}
\toprule
Subject & 1 & 2 & 3 & 4 & 5\\
\midrule
Data identifier & M006 & M007 & M008 & M004 & M005\\
Reconstructions $n_a$ & 6 & 6 & 7 & 6 & 6\\
\bottomrule
\end{tabular}
\end{table}

\subsection{Nonparametric Bootstrap Implementation and Results}
For both landmarks, the observed group extrinsic sample means, covariance estimates, differentials, precision weights, and pooled tangent bases are computed from the 31 original configurations. In the nonparametric bootstrap, complete configurations are resampled with replacement independently within each subject, using identical resampled reconstruction indices for the two landmarks and for the OPS and PS calculations. Only the group extrinsic sample means are recomputed in \eqref{eq:fixedboot}; the observed covariance estimates and tangent bases remain fixed. We use $20{,}000$ resamples and pseudorandom seed $42$. The empirical $97.5\%$ bootstrap quantile gives the marginal Bonferroni critical value; the Monte Carlo upper-tail estimate is $(1+m)/(20{,}001)$ \cite{PhSm:2010}, where $m$ counts bootstrap statistics not smaller than the observed statistic. This add-one formula is a conservative Monte Carlo reporting convention; the justification of the bootstrap itself is the first-order result of Proposition~\ref{prop:boot}.

The parallel VW PS calculation tests equality of population VW mean axes at each landmark. For group $a$, let $K_{a,n_a}=n_a^{-1}\sum_iY_{a,i}Y_{a,i}^\top$ and let $\widehat P_{a,n_a}$ be its leading rank-one eigenprojector. The pooled reference is the leading eigenprojector $\widehat P_n=v_{\mathrm{PS},n}v_{\mathrm{PS},n}^\top$ of $\sum_a(n_a/n)\widehat P_{a,n_a}$, with tangent basis $D_{\mathrm{PS},n}=B(v_{\mathrm{PS},n})$ and Frobenius tangent coordinates $\psi(H)=\sqrt2\,D_{\mathrm{PS},n}^\top Hv_{\mathrm{PS},n}$ for symmetric $H$. The discrepancies are $\psi(\widehat P_{a,n_a}-\widehat P_n)$. Each group covariance estimate is the average outer product of the vectors $\psi\bigl(\widehat R_a^{\VW}H_i\widehat P_{a,n_a}+\widehat P_{a,n_a}H_i\widehat R_a^{\VW}\bigr)$ with $H_i=Y_{a,i}Y_{a,i}^\top-K_{a,n_a}$, where $\widehat R_a^{\VW}$ is the reduced resolvent formed from the eigenvectors and eigenvalue gaps of $K_{a,n_a}$. The statistic $T_n^{\mathrm{PS}}$ is \eqref{eq:weighted} with these quantities, and its fixed-covariance bootstrap recomputes only the group projectors, with errors $\psi(\widehat P_{a,n_a}^*-\widehat P_{a,n_a})$.

\begin{table}[t]
\centering
\caption{Five-group marginal tests of equality of population extrinsic means. OPS and VW PS concern different mean parameters. The nonparametric bootstrap uses fixed observed covariance weights, with $20{,}000$ resamples and Bonferroni level $0.025$ for each landmark.}
\label{tab:tests}
\begin{tabular}{llrrr}
\toprule
Method & Original landmark & Statistic & 97.5\% cutoff & Bootstrap $p$-value\\
\midrule
OPS & 3 (right inner brow) & 141.47 & 23.02 & 0.000050\\
OPS & 6 (chin) & 947.87 & 22.70 & 0.000050\\
\addlinespace
VW PS & 3 (right inner brow) & 141.44 & 22.98 & 0.000050\\
VW PS & 6 (chin) & 966.25 & 22.74 & 0.000050\\
\bottomrule
\end{tabular}
\end{table}

\begin{table}[t]
\centering
\caption{Diagnostics for the observed marginal $3\times3$ extrinsic tangent covariance estimates. The bounds are rounded outward and do not assert nonsingularity in every bootstrap resample.}
\label{tab:diagnostics}
\begin{tabular}{lccc}
\toprule
Method & Matrices & \shortstack{Minimum eigenvalue\\(lower bound)} & \shortstack{Maximum condition no.\\(upper bound)}\\
\midrule
OPS & 10 & $1.44\times 10^{-6}$ & $1.87\times10^3$\\
VW PS & 10 & $2.89\times 10^{-6}$ & $1.87\times10^3$\\
\bottomrule
\end{tabular}
\end{table}

All ten observed OPS group-by-landmark $3\times3$ tangent covariance matrices have full rank; Table~\ref{tab:diagnostics} reports the positive-definiteness bounds for both constructions. The smallest group mean resultant length is $0.9958$, and the smallest leading eigenvalue gap of the $K_{a,n_a}$ is $0.9837$. Table~\ref{tab:tests} reports the observed statistics and fixed-observed-covariance bootstrap results; none of the $20{,}000$ statistics attains or exceeds its corresponding observed value in any of the four calculations. Thus the reported $p$-value $1/20{,}001$ expresses finite simulation resolution, not a zero tail probability.

We reject the null hypothesis of equality of the five population extrinsic OPS mean directions separately at the right inner brow and at the chin under the stated testing model. The parallel VW PS calculations use the projector $u\mapsto uu^\top$ and test equality of population \emph{axes}, not directed OPS means. They are a descriptive methodological comparison and reproduce the same marginal decisions. The OPS two-landmark family and the PS two-landmark family each use their own Bonferroni threshold $0.025$; no four-test joint familywise adjustment is asserted. The similarity of test statistics does not prove that OPS and PS are equivalent or that one procedure has greater general power. The omnibus conclusions do not identify which pairs of subjects have different population means.

\subsection{Descriptive Changes Between Estimated Means}\label{sec:changes}
The quaternion group also permits a descriptive change between two group extrinsic mean estimates. For subject groups $a\ne b$ and landmark $L$, define
\begin{equation}\label{eq:change}
 \widehat c_{ab}^{(L)}=(\bar Y_{a,E,n_a}^{(L)})^{-1}\odot\bar Y_{b,E,n_b}^{(L)}\in\Ss^3.
\end{equation}
This is a change \emph{between estimated means}, not the mean of paired within-subject changes, and it is distinct from the fitted tangent coordinate $\widehat c_n$ of \eqref{eq:weighted}. Writing $\widehat c_{ab}^{(L)}=(\xi,t)$ in the scalar-last convention, its associated principal $SO(3)$ rotation angle is $\widehat\vartheta=2\arccos|t|$. This angle is the same for a quaternion and its antipode, and therefore does not summarize the full signed OPS change; moreover, it is not a measured physical rotation of a face. For chin coordinates, the mean-to-mean changes from Subject~5 to Subjects~1--4 have associated angles $18.38^\circ$, $23.12^\circ$, $17.08^\circ$, and $18.09^\circ$, respectively. The six pairwise angles among Subjects~1--4 are all smaller than $6.4^\circ$, and all observed quaternion changes, at both landmarks, have positive real part. These numerical descriptions do not yield pairwise $p$-values or confidence regions.

\section{Discussion}
The product-of-spheres representation of OPS supplies a direct route from signed projective coordinates to extrinsic mean inference. The global quaternion tangent frame yields explicit covariance blocks, and the fitted generalized least-squares quadratic form has a chi-square null limit when the full population covariances are positive definite. The inverse-free discrepancy has a distinct nonpivotal limit and allows singular covariances, but its role should not be confused with the weighted facial tests reported here.

Three limitations should be noted. First, all OPS computations are relative to the stated projective frame, positive homogeneous lifts, and common fixed-sign domain. An arbitrary alternative projective frame can change the extrinsic metric and the test's numerical value. Second, independent reconstructions are the units assumed by the sampling model; the small group sizes limit the reliability of asymptotic and bootstrap approximations. Third, separate marginal rejections do not amount to a joint two-coordinate MANOVA test, and do not locate significant population pairs. Larger groups or a separately validated singular-covariance joint analysis would be required for that extension.

The principal contribution is a mathematically specified, reproducible framework for comparing population extrinsic OPS mean directions, with explicit treatment of coordinate signs, the full covariance structure, and nonparametric resampling conventions.

\section*{Acknowledgment, Data and Software}
Paige and Patrangenaru thank the National Science Foundation for awards NSF-DMS:2311058 and NSF-DMS:2311059, respectively.
The 31-reconstruction landmark file and the analysis program that produces the reported results are available at \url{https://socrates.stat.fsu.edu/~vic/thesis-computation-Musab}.

\end{document}